\documentclass[12pt,a4paper]{article}
\pdfoutput=1
 \usepackage{jheppub}
 \usepackage[utf8]{inputenc}
 \usepackage{amsmath,amsthm}
 \usepackage{amsfonts}
 \usepackage{mathtools}
 \usepackage{amssymb}
 \usepackage{caption}
 \usepackage{subcaption}
 \usepackage{slashed}
\usepackage{float}
\usepackage{graphicx}
\graphicspath{ {project2 tex/} }

\newcommand{\be} {\begin{equation}}
\newcommand{\ee} {\end{equation}}
\newcommand{\bea} {\begin{eqnarray}}
\newcommand{\eea} {\end{eqnarray}}
\newcommand{\ba} {\begin{array}}
\newcommand{\ea} {\end{array}}
\newcommand{\nn} {\nonumber}

\theoremstyle{plain}

\newtheorem{lem}{Lemma}
\newtheorem{prop}{Proposition}

\theoremstyle{definition}

\newtheorem{theorem}{Theorem}

 \title{Sign of BPS index for ${\cal N}=4$ dyons}
 \author{Aradhita Chattopadhyaya}
\affiliation{School of Theoretical Physics\\
Dublin Institute for Advanced Studies\\
Dublin 4, Ireland}
\emailAdd{aradhita@stp.dias.ie}

\abstract{In this paper we argue how the sign changes on an average for the positive weight mock modular forms associated with the ${\cal N}=4$ type II string black holes compactified on orbifolds of $K3\times T^2$. The orbifolds of order $N$ act with $g'\in[M_{23}]$ an order $N$ symplectic orbifold on $K3$ and a $1/N$ shift in one of the circles of the torus $T^2$. We expand the inverse Siegel modular forms of subgroups of $Sp_2(\mathbb{Z})$  for the magnetic charge $P^2=2$ in terms of mock Jacobi forms and Appell Lerch sums. We analyze the average growth of the coefficients of these mock modular forms after theta decomposition and removing inverse eta products. In particular we remove the contribution of the fundamental string which rightfully dominates the growth of the positive weight modular forms after the first few coefficients and ensures the positivity of the helicity trace index $-B_6$. Using numerics and limits of divisor sum function we predict the sign of these mock modular forms. We also observe that the cusp forms associated with the non-geometric orbifolds of $K3$ can only contribute for sign changes up to the first few terms hence their contribution can be neglected for large electric charges.
}

\begin{document}
\maketitle
\flushbottom

\section{Introduction}
Black hole microstates counting and extracting the Bekenstein Hawking entropy is a subject of great interest \cite{Strominger:1996sh,Dijkgraaf:1996it}.
Partition function of ${\cal N}=4$ supersymmetric black holes in type II string theory were studied in \cite{Dijkgraaf:1996it}. In later works partition function for these black holes for geometric orbifolds on $K3$ (CHL orbifolds) were studied \cite{Jatkar:2005bh,LopesCardoso:2004law,David:2006yn,Sen:2007qy,Dabholkar:2012zz}. The calculation is done by computing the helicity trace index $B_{2n}$, where $4n $ is the number of broken supersymmetry generators for these dyons \cite{Sen:2010mz}. 
The helicity trace index $B_{6}$ computes the degeneracy of the 1/4-th BPS dyons in these theories.
\be
B_{2n}=\frac{1}{(2n)!}{\rm Tr}\left(e^{2\pi i h}(2h)^{2n}\right)
\ee
where $h$ is the third component of the angular of the black hole in its rest frame. The trace is over all states with a given set of electric and magnetic charges.

In ${\cal N}=4$ supersymmetric theories for the $1/4$-th BPS dyons the number of broken supersymmetry generators is 12. These correspond to hair modes for the black holes \cite{Sen:2010mz}. In \cite{Sen:2010mz} it was assumed that there is a frame of reference where the only angular momentum in the rest frame comes from these broken supersymmetry generators.
The number of microstates of the black hole can be counted from $-B_6$
\be
B_6=-d_{\rm horizon}\times d_{\rm hair}
\ee
where $d_{\rm horizon}$ counts the number of degrees of freedom for a given set of charges while $d_{\rm hair}$ counts the hair degrees of freedom for the BPS black holes.
In \cite{Banerjee:2009uk,Jatkar:2009yd} it was observed that these hair modes were frame dependent and in one reference frame this was computed for the canonical case of $K3\times T^2$. In \cite{Chattopadhyaya:2020yzl} it was found that not in all ${\cal N}=4$ type II black holes these fermionic zero modes are only hairs. Explicit computation of the hair modes were done in \cite{Chakrabarti:2020ugm,Chattopadhyaya:2020yzl}. The interesting feature we found is that these hairs are always given by partition functions of some bosonic degrees of freedom.

In this paper we argue the positivity of the generating function of BPS indices of single centered ${\cal N}=4$ supersymmetric black holes in type II theories. These generating functions are given in terms of inverse Siegel modular forms of subgroups of $Sp_2(\mathbb{Z})$ \cite{Dabholkar:2005by,David:2006ji,David:2006ru,David:2006ud,Govindarajan:2009qt,Sen:2007vb,Persson:2013xpa,Chattopadhyaya:2017ews}. At fixed magnetic charges we can expand these inverse Siegel forms in terms of Jacobi forms. This was done for $K3$ case in \cite{Dabholkar:2012nd} at small values of magnetic charge. A list of these Jacobi forms were obtained for the orbifold compactifications for magnetic charge $P^2=2$ in \cite{Chattopadhyaya:2018xvg}. These Jacobi forms contain poles which can be written as Appell Lerch sums and a finite piece. In this paper we use the results of \cite{Dabholkar:2012nd,Chattopadhyaya:2018xvg} to see the behavior of the relevant Jacobi forms after removing their poles and also extracting out the contribution of the fundamental string. In particular, for most of our analysis we take the corresponding weight $k=9/2$ mock modular form which is obtained after theta decomposition and extracting out an extra factor of $\frac{1}{\eta^6}$ such that the resultant function has no effect from any bosonic partition function whose growth is $\sim e^{\sqrt{n}}$ while the coefficients of modular forms with positive weight can only grow polynomially.

This paper is organized in the following manner. In section \ref{bpsreview} we give a brief review of computing $1/4$-th BPS dyons. In section \ref{modforms} we give the definitions of modular, Jacobi, mock modular, mock Jacobi forms. In section \ref{props} we propose two statements based on numerical evidence and assuming a similarity with Prime number theorem. We also review some theorems on growth of the sum of divisors of $n$. In section \ref{signchange} we review the results of the canonical case of $K3\times T^2$, in section \ref{neq2} we extend the results to the 2A orbifold based on numerical evidence and the results of section \ref{numth} and in section \ref{oddn} we further extend the results for orbifolds of order $N>2$ odd.

\section{Counting $1/4$-th BPS dyons}\label{bpsreview}
In this section we give a brief introduction of the computation of the 6th helicity trace index, which computes the number of $1/4$th BPS dyons in ${\cal N}=4$ supersymmetric type II theories. We consider the compactifications on symplectic orbifolds of $K3$ by an orbifold of $g'$ of order $N$ corresponding to conjugacy classes of Mathieu group $M_{23}$. In the attractor chamber only single centered solutions survive and since the black hole is spherically symmetric and regular at the horizon and assuming there is a frame of reference where the only hair modes were the 12 broken supersymmetry generators the sign of $B_6$ was predicted to be negative \cite{Sen:2010mz}. We can obtain the Fourier coefficient of the black hole partition functions for given electromagnetic charges in the attractor chamber. This can be written as follows:
\begin{eqnarray}\label{degen}
d(Q, P)  = \frac{1}{N} ( -1)^{ Q\cdot P +1} 
\int_{{\cal C}} d\tilde\rho d\tilde\sigma d \tilde v \; 
e^{-\pi i ( N \tilde\rho Q^2 + \tilde\sigma P^2/N + 2 \tilde v Q\cdot P ) }
\frac{1}{ \tilde \Phi_k ( \tilde \rho, \tilde \sigma , \tilde v) }.
\end{eqnarray}
The contour ${\cal C}$ is defined over a 3 dimensional subspace of $\mathbb{C}^3$ where, $$\tilde \rho = \tilde\rho_1 + i \tilde\rho_2, \tilde\sigma = 
\tilde\sigma_1 + i \tilde\sigma_2,
\tilde v = \tilde v_1 + i \tilde v_2  .$$ 
\begin{eqnarray}\label{contour}
\tilde\rho_2 = M_1, \qquad \tilde\sigma_2 = M_2, \qquad \tilde v_2 = - M_3, \\ \nonumber
0\leq \tilde\rho_1 \leq 1, \qquad 0 \leq \tilde\sigma_1 \leq N, \qquad 0 \leq \tilde v_1 \leq 1.
\end{eqnarray}
Here $M_1, M_2, M_3$ are  positive numbers, which are fixed and large and
$M_3 << M_1, M_2$. For CHL models there is a symmetry in the electric and magnetic charges which essentially render the contour as  the Fourier expansions are done in the order $e^{2\pi i \tilde \sigma}$  and $e^{2\pi i \tilde \rho}$ and then in $ e^{-2\pi i \tilde v}$.  The function
$\tilde \Phi$, which is present in the integrand is a Siegel modular form of weight $k$ transforming under a subgroup of $Sp(2, \mathbb{Z})$. We can write this explicitly as:
\begin{eqnarray}\label{siegform}
\tilde{\Phi}_k(\rho,\sigma,v)&=&e^{2\pi i(\tilde\rho+ \tilde\sigma/N+ \tilde v)}\\ \nn
&&\prod_{r=0}^{N-1}
\prod_{\begin{smallmatrix}k'\in \mathbb{Z}+\frac{r}{N},l\in \mathcal{Z},\\ j\in 2\mathbb{Z}\\ k',l\geq 0, \; j<0\;  k'=l=0\end{smallmatrix}}
(1-e^{2\pi i(k'\sigma +l\rho+jv)})^{\sum_{s=0}^{N-1}e^{2\pi isl/N}c^{r,s}(4k'l-j^2)}.
\end{eqnarray}
Here $c^{(r,s)}$ can be obtained from the Fourier coefficients of the twisted elliptic genus of $K3$ given as $F^{(r,s)}(\tau,z)$
\begin{eqnarray}
F^{(r,s)}(\tau,z) &=& \sum_{n,j}c^{(r,s)} (4n-j^2)q^n y^j
\end{eqnarray}
A list of $F^{(r,s)}$ can be found in \citep{Chattopadhyaya:2017ews}.
In general we can write $F^{(r,s)}$ as,
\begin{eqnarray}
F^{(r,s)}(\tau,z) &=& \alpha_N^{(r,s)}A(\tau,z)+\beta_N^{(r,s)}(\tau)B(\tau,z)
\end{eqnarray}
where $\alpha_N^{(r,s)}$ is a positive number for the orbifolds in question. $\beta_N^{(r,s)}$ is generally a function of weight two holomorphic Eisenstein series and cusp forms involving various products of eta functions. 
We can extract the Eisenstein series and cusp forms from $\beta_N^{(r,s)}(\tau)$ $$\beta_N^{(r,s)}(\tau)=\sum_{d|N}f_d{\cal E}_d(\frac{a\tau+b}{d})+b_N \eta_N(\tau)$$
where $f_d,b_N$ are numbers, $\eta_N$ is a weight two cusp form and $a,b$ depend on $N$ and $(a,d)=1$.

We note that the Fourier expansion in terms of Jacobi forms after removing the poles in $v$-plane we get an overlap with the attractor chamber mentioned in \cite{Sen:2010mz}.  In this region we have the classical area of the black hole horizon to be positive and the electromagnetic charges are positive $$P^2Q^2>(Q\cdot P)^2,\quad P^2>0,\quad Q^2>0,\quad Q\cdot P\ge 0.$$
For orbifolds of $K3$ the resultant attractor chamber is smaller than obtained from the above constraints alone, however for this analysis we require only the above mentioned constraints which is also obtained by removing the poles in $v$-plane.

Before moving into the mathematical results and numerical estimates let us briefly state the results of the contributions from the two forms in the heterotic counting for $K3$ and the orbifold models. This was done in \cite{Chattopadhyaya:2020yzl}. In the CHL models the 2-forms compute the contribution of hair that comes from the bosonic modes. These are given by the contribution of the fundamental string however without the shift in zero point energy
\begin{eqnarray}
Z_{{\rm hair}: {\rm CHL}}^{\rm bosonic}&=&\prod_{l=1}^{\infty}(1-e^{2\pi i (Nl\rho)})^{-4}(1-e^{2\pi i (l\rho)})^{-\sum_{s=0}^{N-1} e^{-2\pi i sl/N}c^{(0,s)}(0)}.
\end{eqnarray}
In our analysis we will remove all these states as the growth of these objects are as that of partitions which remain far higher than the growth of ordinary positive weight modular forms.

\section{A brief review of some results in number theory}\label{numth}
Before moving into the details of the calculations we first introduce the notion of modular forms and its associated objects. These are a subject of interest for black hole thermodynamics in string theory as well as counting different kind of partitions in number theory.

\subsection{Modular, Jacobi, mock modular forms}\label{modforms}
In this section we briefly review the definitions of modular forms, Jacobi forms, mock modular forms and mock Jacobi forms. For more rigorous definitions see \cite{Eichler:1985,MR2605321,ZwegersThesis,Dabholkar:2012nd,Manschot:2017xcr,Chattopadhyaya:2023aua}.

\vspace{0.3cm}
\noindent
{\bf Definition:} A modular form $f(\tau)$ is a holomorphic function of $\tau\in\mathbb{H}$ of weight $k$ with the following 
\begin{enumerate}
\item Transformation property:
 \be
f(\frac{a\tau+b}{c\tau+d})=(c\tau+d)^k f(\tau)
\ee
where $\left(\begin{smallmatrix}
 a & b\\
 c & d
\end{smallmatrix}\right)\in \Gamma\subseteq SL_2(\mathbb{Z})$.

\item Boundedness criteria: 
\be
  \label{cuspgrowth}
\lim_{\tau\to i\infty} (c\tau+d)^{-k} f\left(\frac{a\tau+b}{c\tau+d}\right)
\ee
remains bounded for all $\begin{pmatrix}
a & b\\ c & d \end{pmatrix}\in SL_2(\mathbb{Z})$. 
\end{enumerate}
The modular group is generated by $S$ and $T$ transformations given by
\be
S=\begin{pmatrix}
0 & -1 \\ 1 & 0
\end{pmatrix}, \qquad T=\begin{pmatrix}
1 & 1 \\ 0 & 1
\end{pmatrix}.
\ee
Due to the $T$ transformations we can write these in terms of Fourier series with $q=e^{2\pi i\tau}$
\begin{equation}
f(\tau)=\sum_{n=0}^{\infty}a_n q^n.
\end{equation}
The above definition can be easily generalized to vector valued modular forms and modular forms of subgroups of $SL_2(\mathbb{Z})$. An important subgroup of $SL_2(\mathbb{Z})$ is $\Gamma_0(N)$ given by
$$\begin{pmatrix}
a & b\\ Nc & d
\end{pmatrix},\quad {\rm with}\; (a,b,c,d)\in\mathbb{Z}^4$$.

\vspace{0.3cm}
\noindent
{\bf Definition:} Cusp form $f(\tau)$ is a modular forms of positive weight $k$ given by the additional growth constraint
\be
  \label{cuspgrowth}
\lim_{\tau\to i\infty} (c\tau+d)^{-k} f\left(\frac{a\tau+b}{c\tau+d}\right)=0
\ee
for all $\begin{pmatrix}
a & b\\ c & d \end{pmatrix}\in SL_2(\mathbb{Z})$.
The implies that the zeroth order Fourier coefficient of a cusp form is always zero.

The space of modular forms are spanned by Eisenstein series and cusp forms.  
\begin{eqnarray}
E_4(\tau) &=& 1+240\sum_{n}\sigma_3(n) q^n, \quad \sigma_\ell(n)=\sum_{d|n}d^\ell, \quad q=e^{2\pi i\tau}\\ \nn
{\cal E}_N(\tau) &=& \frac{1}{N-1}(NE_2(N\tau)-E_2(\tau))
\end{eqnarray}
where $E_2(\tau)=1-24\sum_{n=1}^{\infty} \sigma(n)q^n$. The weight two Eisenstein series is mock modular (or quasi-modular) but ${\cal E}_N$ is a holomorphic modular form of $\Gamma_0(N)$.

\noindent
Modular forms are generally defined with a boundedness criteria, however those do not hold everywhere in the upper half of the complex plane if the weight $k$ is negative.

\vspace{0.3cm}
\noindent
{\bf Definition:} A weight $k$ index $m$ Jacobi form $J(\tau,z)$ is a holomorphic function of $\tau\in\mathbb{H}$  and $ z\in \mathbb{C}$ with the following transformation properties:
\begin{enumerate}
\item For $\left(\begin{smallmatrix} a & b\\ c & d \end{smallmatrix} \right)\in SL_2(\mathbb{Z})$ $$J\left(\frac{a\tau+b}{c\tau+d},\frac{z}{c\tau+d}\right) = (c\tau+d)^k e^{\frac{2\pi imcz^2}{c\tau+d}} J(\tau,z)$$
\item For $(\lambda,\mu)\in\mathbb{Z}^2$ 
$$J(\tau,z+\lambda\tau+\mu)= e^{2\pi im(\lambda^2\tau+2\lambda z)}J(\tau,z)$$
\item $J(\tau,z)$ has a Fourier expansion:
$$ J(\tau,z)=\sum_{n\ge 0}\sum_{\ell^2\le 4mn} c_{n,\ell} q^{n} e^{2\pi iz\ell}$$
\end{enumerate}
Jacobi forms of index 1 and weight $k\in\mathbb{Z}$ $J(\tau,z)$ can be decomposed with respect to theta functions also known as theta decomposition.
$$J(\tau,z)=\theta_3(2\tau,2z)j_0(\tau)+\theta_2(2\tau,2z)j_1(\tau)$$
where $j_0$ and $j_1$ are weight $k-1/2$ modular forms.
In the case of weakly holomorphic Jacobi forms of index $1$ as the twisted elliptic genus of $K3$ the Fourier expansion of $4n-\ell^2\ge -1$ instead of $4n\ge 0$.
$$ J(\tau,z)=\sum_{n\ge 0}\sum_{\ell^2\le 4n+1} c_{n,\ell} \;q^{n} e^{2\pi iz\ell}.$$
The theta decomposition is still possible for the weakly holomorphic cases.

\vspace{0.3cm}
\noindent
{\bf Definition:} A mock modular form $f(\tau)$ of weight $k$ is a holomorphic function in $\tau\in\mathbb{H}$ such that its completion $\hat f(\tau,\bar\tau)=f(\tau)+g^*(\tau,\bar\tau)$ transforms like a modular form of weight $k$ where $g^*$ is given as follows:
\begin{eqnarray}
  \label{gstar}
g^*(\tau,\bar{\tau})=-2^{1-k}i\int_{-\bar{\tau}}^{i\infty} \frac{g(-v)}{(-i(v+\tau))^{k-\ell}}\, dv\;.
\end{eqnarray}
and $g(\tau)$ is a modular form of weight $2-k$.

\vspace{0.3cm}
\noindent
A mock Jacobi form $J(\tau,z)$ of weight $k$ and index $1$ is given by a similar theta decomposition property as a Jacobi form of index 1, however in this case,
$$J(\tau,\nu)=j_0(\tau)\theta_3(2\tau,2z)+j_1(\tau)\theta_2(2\tau,2z)$$
where $j_0,j_1$ are mock modular in nature.

\subsection{Growth of $\sigma(n)$- theorems and numerical evidence}\label{props}
Our argument involves two theorems from number theory and one assumption. Before stating the assumption we first write the statements of the theorems we use. The first one is a statement on the maximal growth of the divisor sum function ie,
$\sigma(n)=\sum_{d|n}d$ for positive $d,n$ in $\mathbb{Z}$. The upper limit on $\sigma(n)$ was given by Ramanujan assuming Riemann hypothesis.
The upper bound can be written as
\begin{eqnarray}
\sigma(n)<e^{\gamma}n\log\log(n)
\end{eqnarray} 
where $\gamma=2.718$ is the Euler Mascheroni constant. Robin showed that considering Riemann hypothesis being false this limit has some modifications \cite{robin}. We state Robin's theorem.

\vspace{0.3cm}
\noindent
\begin{theorem}\label{robins}
(Robin's Theorem) Considering Riemann hypothesis to be false, the upper limit of divisor sum $\sigma(n)$ for $n>3$ can be given by,
\begin{eqnarray}
\sigma(n)<e^{\gamma}n\log\log(n)+\frac{0.6483n}{\log\log(n)}<e^{\gamma}n\log\log(n)+O(n)
\end{eqnarray} 
\end{theorem}

Most of our analysis revolves around the average growth of the divisor sum and we need to estimate the average growth of $\sigma(n)=\sum_{d|n}d$.

\vspace{0.3cm}
\noindent
\begin{theorem}\label{avg}
Average growth of $\sigma(n)\sim \frac{\pi^2 n}{6}+O(\log(n)).$ It is important to note that the average growth of the $\sigma(n)$ function is linear and close to $1.645n<2n$ \cite{hardy-wright} \footnote{see theorem 324}.
\end{theorem}

\vspace{0.3cm}
\noindent
Next we state the prime number theorem. This theorem will not be used in the paper but we will motivate a proposition based on numerical evidence which relates closely to prime number theorem.

\vspace{0.3cm}
\noindent
\begin{theorem}\label{pnt}
{(Prime Number Theorem)}
Let $\pi(n)$ denote the number of primes below $n$, then as $n\rightarrow\infty$, $$\pi(n)=\frac{n}{\log(n)}.$$
A simple proof of the prime number theorem given by Newman can be found in \cite{zagier-prime}.
\end{theorem}

\noindent
Motivated by the above results we propose the following based on numerical evidences. We introduce two propositions for the growth of the following ratio
$\frac{\sigma(4n+k}{2\sigma(2n)-4\sigma(n)}$.

\vspace{0.3cm}
\noindent
\begin{prop}\label{prop1p}
For a positive integer $n>10^5$ the following ratios
\begin{eqnarray}\label{prop1}
\frac{\sigma(4n+1)}{2\sigma(2n)-4\sigma(n)}<1, \quad \frac{\sigma(4n+3)}{2\sigma(2n)-4\sigma(n)}<1.
\end{eqnarray}
hold for a maximum of $\frac{\kappa n}{\log(n)}$ terms and $\kappa\le 2$. 
\end{prop}

\vspace{0.3cm}
\noindent
Numerically this is checked to be correct for $10^8$ terms.
We shall not attempt to prove this in this paper. A heuristic argument can be given as follows:
The possibility of finding $\sigma(4n+1)<2\sigma(2n)-4\sigma(n)$ occurs most often when $4n+1$ is a prime number (or a product of very few terms) and $n$ contains several factors such that $2\sigma(2n)-4\sigma(n)>4n$. A similar possibility also arises for $4n+3$ primes, however as $n$ cannot have the divisor 3, it is even harder to find such numbers for small $n$. Upto $10^8$ there is only one term where $\sigma(4n+3)<2\sigma(2n)-4\sigma(n)$ \footnote{We checked upto $n=10^8$ for the ratios $\frac{\sigma(4n+1)}{2\sigma(2n)-4\sigma(n)}$. There are $1.4\times 10^6$ terms where the ratio is below 1 and number of these terms are below less than the number primes below $10^8$ which is approximately $5.4\times 10^6$.} which is clearly far below the proposed bound $\frac{\kappa n}{\log(n)}$ with $\kappa\le 2$.

\vspace{0.3cm}
\noindent
\begin{prop}\label{prop2p}
For a positive integer $n>10^5$ the following also holds
\begin{eqnarray}\label{prop2}
\sigma(4n+1)-(2\sigma(2n)-4\sigma(n))<n, \quad \sigma(4n+3)-(2\sigma(2n)-4\sigma(n))<n.
\end{eqnarray}
for a maximum of $\frac{\kappa n}{\log(n)}$ terms upto $n$ and $\kappa\le 2$. 
\end{prop}

\vspace{0.3cm}
\noindent
Numerically this is also checked to be correct for $10^8$ terms.
We shall not attempt to prove this in this paper. In the following sections we will see the need for the above propositions and the theorem on average growth of $\sigma(n)$.

\section{Sign changes for compactifications on orbifolds of $K3\times T^2$}\label{signchange}
In this section we introduce the behavior of the Fourier coefficients of the inverse Siegel modular forms associated with the $1/4$-th BPS dyons of ${\cal N}=4$ type IIB theories compactified on $K3\times T^2$ and its orbifolds. We then extract out the hair modes as obtained in \cite{Chakrabarti:2020ugm,Chattopadhyaya:2020yzl} and theta decompose to observe the positivity of the Fourier coefficients of the resultant half integral weight modular forms for the magnetic charge $P^2=2$.
The inverse Siegel modular form of weight $k$ can be expanded in terms of weak Jacobi forms $A$ and $B$ given as:
\begin{eqnarray}
\frac{1}{\Phi_{k}(q,p,y)} &=& \sum_{m=-1}^{\infty}\psi^{(N)}_m(q,y)p^{m/N}
\end{eqnarray}
where the order of $K3$ orbifold is given by $N$.
The Jacobi forms associated to the canonical case for the compactification $K3\times T^2$ at magnetic charge $P^2=2$ can be given as follows:
\begin{eqnarray}
\psi_1^{(1)} &=& -48\frac{A^2(\tau,z)}{B(\tau,z)} 
\end{eqnarray}
where, the weak Jacobi form of weight zero and weight $-2$ given as $A,B$ respectively (index 1 for both cases) are defined as follows:
\begin{eqnarray}\nn
A(\tau,z) &=& \frac{1}{\eta^6}\left[6\theta_3^2(\tau,z)\theta_2^2\theta_4^2 - 2\theta_1^2(\tau,z)(\theta_4^4-\theta_2^4))\right], \quad B(\tau,z)=\frac{\theta_1^2(\tau,z)}{\eta^6}
\end{eqnarray}
$\psi_1^{(N)}$ can be written in terms of an Appell Lerch sum which is a sum over poles in the $z$ plane and a mock modular form which remains finite  written as follows:
\begin{eqnarray}
\psi_1^{(N)} &=& \psi_1^{(N)}|_{\rm polar}+\psi_1^{(N)}|_{\rm finite}
\end{eqnarray}
For the canonical case of $K3$ we get,
\begin{eqnarray}
\psi_1^{(1)}|_{\rm polar}=\frac{48}{\eta^{24}}\sum_{n\in\mathbb{Z}}\frac{q^{n^2+n}y^{2n+1}}{(1-q^ny)^2},\quad q=e^{2\pi i\tau}, \;\;y=e^{2\pi iz}.
\end{eqnarray}
The finite peice can be given in terms of class number generating function and the Eisenstein series $E_4$
\begin{eqnarray}
\psi_1^{(1)}|_{\rm finite}=-\frac{3}{\eta^{24}}(E_4(\tau)B(\tau,z)+216{\cal H}(\tau,z)
\end{eqnarray}
We can write ${\cal H}(\tau,z)$ in terms of Hurwitz class number generating functions. It is possible to theta decompose $B(\tau,z)$ and ${\cal H}(\tau,z)$ which separate the even and odd powers of $y$, given as
\begin{eqnarray}
B(\tau,z) =\frac{\theta_1^2(\tau,z)}{\eta^6}=b_0(\tau)\theta_3(2\tau,2z)+b_1(\tau)\theta_2(2\tau,2z)\\ \nn
{\cal H}(\tau,z)=h_0(\tau)\theta_3(2\tau,2z)+h_1(\tau)\theta_2(2\tau,2z)
\end{eqnarray}
where we have,
\begin{eqnarray}
b_0(\tau)=\frac{1}{\eta^6}(\theta_2(2\tau)),\quad b_1(\tau)=\frac{1}{\eta^6}(-\theta_3(2\tau))
\end{eqnarray}
One can obtain an extensive list of the Fourier coefficients of $h_0(\tau)$ and $h_1(\tau)$ referred to as $f_{2,0}$ and $f_{2,1}$ respectively in \cite{Chattopadhyaya:2023aua} table 1.
It was shown in \cite{Bringmann:2012zr,Chattopadhyaya:2020yzl} that the Fourier coefficients of $\psi_1^{K3}$ after theta decomposition can only change sign at $q^0$ and the sign remains correct as predicted from black hole physics even after removing the contribution of the hair modes. Since $\eta^6(\tau)$ is a cusp form  one can similarly argue that the weight $9/2$ mock modular forms obtained as $h_{\mu}\eta^6$ will have coefficients growing $<a_{\Delta}(\Delta)^{-9/4}n^{9/4}$ and thus,
\begin{eqnarray}
-3\eta^6(\tau)(E_4(\tau)b_0(\tau)+216h_0(\tau),\quad 3\eta^6(\tau)(E_4(\tau)b_1(\tau)-216h_1(\tau)
\end{eqnarray}
can only have a finite number of sign changes (upto $q^2$).
In the rest of this paper we shall mostly ignore cusp forms as their saddle point growth remains $a_{\Delta}\Delta^{-k/2}n^{k/2}\sim O(n^{k/2})$ and from Deligne bound \cite{deligne:74,deligne:80} these have growth $O(n^{\frac{k-1}{2}})$ while the modular forms have the saddle point growth as $O(n^{k})$ and average growth $O(n^{k-1})$. So if there is a sign change this will only affect the first few terms which can be easily computed numerically.

\subsection{2A orbifold}\label{neq2}
Let us now write the explicit forms of $\psi_1^{(2)}$ for the 2A orbifold of $K3$ where the orbifold is a $g'$ action corresponding to $[M_{23}]$, where $M_{23}\subset M_{24}$ Mathieu moonshine group. 
In this case we need the definition of the holomorphic Eisenstein series of weight 2
given by ${\cal E}_N$ 
\begin{equation}
{\cal E}_N(\tau) = \frac{1}{N-1}\left(N E_2( N\tau)-E_2(\tau)\right).
\end{equation}
We can write the finite part as $g^8(\tau)\psi_1^{(2)}|_{finite}$ as:
\be
-\frac{20}{9}A(2\tau,z){\cal E}_2(\tau)-B(2\tau,z)\left(\frac{1}{16}E_4(\tau)+\frac{13}{36}E_4(2\tau)+\frac{19}{144}{\cal E}_2^2(\tau)\right)-104{\cal H}(2\tau,z)
\ee
where $g^8(\tau)=\eta^8(\tau)\eta^8(2\tau)$.
Simplifying using ${\cal E}_2^2(\tau)=\frac{4}{5}(E_4(2\tau)+\frac{1}{4}E_4(\tau))$ we get,
\be
-\frac{20}{9}A(2\tau,z){\cal E}_2(\tau)-B(2\tau,z)\left(\frac{4}{45}E_4(\tau)+\frac{21}{45}E_4(2\tau)\right)-104{\cal H}(2\tau,z)
\ee
The components $\frac{4}{45}E_4(\tau)+\frac{21}{45}E_4(2\tau)$ will have always positive Fourier coefficients. So does ${\cal E}_2(\tau)$ has only positive coefficients  in its Fourier expansion as $\sigma(2n)>2\sigma(n)$.

The last two terms are similar to the results obtained in the previous section and the only term which we need to deal with in this part is $\phi^{2A}_A=-A(2\tau,z){\cal E}_2(\tau)\eta^6$. 
The theta decomposition for $A(\tau,z)\eta^6$ can be given as follows:
\begin{eqnarray}\nn
A(\tau,z)\eta^6 &=& 6\theta_3^2(\tau,z)\theta_2^2\theta_4^2 - 2\theta_1^2(\tau,z)(\theta_4^4-\theta_2^4)=A_0(\tau)\theta_3(2\tau,2z)+A_1(\tau)\theta_2(2\tau,2z)
\end{eqnarray}
where we have,
\begin{eqnarray}\nn
A_0(\tau) &=&  6\,\theta_2^2\theta_4^2\theta_3(2\tau)-2(\theta_4^4-\theta_2^4)\theta_2(2\tau),\quad A_1(\tau) =6\,\theta_2^2\theta_4^2\theta_2(2\tau)+2(\theta_4^4-\theta_2^4)\theta_3(2\tau) \\ \label{a0a1}
\end{eqnarray}
We used the double angle identities to do the theta decomposition above,
\begin{eqnarray}
\theta_1^2(\tau,z)&=&\theta_2(2\tau)\theta_3(2\tau,2 z)-\theta_3(2\tau)\theta_2(2\tau,2 z)\\ \nn
\theta_3^2(\tau,z)&=&\theta_3(2\tau)\theta_3(2\tau,2 z)+\theta_2(2\tau)\theta_2(2\tau,2 z)
\end{eqnarray}

\noindent
Before moving into the relations between the divisor sums and the propositions let us first check the asymptotic signs from saddle point analysis. For this we first need the result of various S transformations involving theta functions:
\begin{eqnarray}
\theta_2(-\frac{1}{\tau}) &=& \sqrt{-i\tau}\theta_4(\tau)\\ \nn
\theta_3(-\frac{1}{\tau}) &=& \sqrt{-i\tau}\theta_3(\tau)\\ \nn
\theta_4(-\frac{1}{\tau}) &=& \sqrt{-i\tau}\theta_2(\tau)
\end{eqnarray}
Also for the multiple angled theta functions we can write,
\begin{eqnarray}
\theta_2(-\frac{2}{\tau}) &=& \frac{1}{\sqrt{2}}\sqrt{-i\tau}(\theta_3(2\tau)-\theta_2(2\tau))\\ \nn
\theta_3(-\frac{2}{\tau}) &=& \frac{1}{\sqrt{2}}\sqrt{-i\tau}(\theta_3(2\tau)+\theta_2(2\tau))
\end{eqnarray}
Now using these we can write the transformation properties of the even and odd parts of theta decomposition of $\eta^6A(\tau,z)$ given by $A_0$ and $A_1$. This is given by
\begin{eqnarray}
A_{0}(\frac{1}{\tau}) &=& \frac{1}{\sqrt{2}}(-i\tau)^{5/2}(A_0(\tau)+A_1(\tau))\\ \nn
A_{1}(\frac{1}{\tau}) &=& \frac{1}{\sqrt{2}}(-i\tau)^{5/2}(A_0(\tau)-A_1(\tau))
\end{eqnarray}
for $A_0,A_1$ as in (\ref{a0a1}).
Therefore we can write
\begin{eqnarray}
A_0(\tau) &=& \frac{1}{\sqrt{2}}(-i\tau)^{-5/2}\left(A_0(-\frac{1}{\tau})+A_1(-\frac{1}{\tau})\right)\\ \nn
A_1(\tau) &=& \frac{1}{\sqrt{2}}(-i\tau)^{-5/2}\left(A_0(-\frac{1}{\tau})-A_1(-\frac{1}{\tau})\right)
\end{eqnarray}

Now we do a saddle point analysis for $A_j$, consider $\tau_2\rightarrow 0$ and $j$ corresponds to 0 or 1. Let us denote the coefficient of the Fourier expansion of $A_j(\tau)=\sum_{n}a_j(n)q^n$ for $n\in\mathbb{Z}+1/4-j/4$
\begin{eqnarray}
a_{j}(n) &=& \int_{0}^1 e^{-2\pi i n\tau}A_j(\tau)d\tau
\end{eqnarray}
Now we have,
\begin{equation}
a_{j}(n) = \frac{(-i)^{5/2}}{\sqrt{2}}\int_{0}^1 e^{-2\pi i n\tau}\tau^{-5/2}\sum_{\ell=0}^1\chi_{\ell,j}A_\ell(-\frac{1}{\tau})d\tau
\end{equation}
for the characters $\chi_{ij}$ as follows:
 $$\chi_{0,0}=1,\quad \chi_{0,1}=1,\quad \chi_{1,0}=1,\quad \chi_{1,1}=-1 $$.
The $q^0$ term appears only in $A_1$ whose coefficient is 2 which is the major contribution from the inverted forms evaluated as $\frac{1}{\tau}\rightarrow \infty$.
At saddle point obtained from the extremization of the exponent $\tau_s=\frac{ik}{2\pi n}$ for a weight $k$ modular form the rough growth estimates as
\begin{eqnarray}
a_0(n)\sim \sqrt{2} n^{5/2}(4\pi e/5)^{5/2}\\ \nn
a_1(n)\sim -\sqrt{2} n^{5/2}(4\pi e/5)^{5/2}
\end{eqnarray}
Now $a_0$ corresponds to even $Q\cdot P$ while $a_1$ corresponds to odd values of $Q\cdot P$. Hence the sign asymptotically looks as expected from $-B_6$ considering the overall sign of the function $A(\tau,z){\cal E}_2(\tau)$. 

\noindent
Next we attempt to use the propositions for the growth of the divisor sum.
Using the weight 2 Eisenstein series $E_2=1-24\sum_{n=1}^{\infty}\sigma(n)q^n$ we can write $A_0,A_1$ as follows:
\begin{eqnarray}\nn
A_0 &=& 6\theta_3(2\tau)\left[-\frac{i}{6}E_2(\frac{\tau-1}{4})+\frac{i}{2}E_2(\frac{\tau-1}{2})-\frac{i}{3}E_2(\tau)\right]-2\theta_2(2\tau)\left[2E_2(\tau)-E_2(\frac{\tau-1}{2})\right]\\ \nn
A_1 &=& 6\theta_2(2\tau)\left[-\frac{i}{6}E_2(\frac{\tau-1}{4})+\frac{i}{2}E_2(\frac{\tau-1}{2})-\frac{i}{3}E_2(\tau)\right]+2\theta_3(2\tau)\left[2E_2(\tau)-E_2(\frac{\tau-1}{2})\right]\\ \nn
\end{eqnarray}
We can also write the above in terms of holomorphic Eisenstein series:
\begin{eqnarray}\nn
A_0 &=& i\theta_3(2\tau)\left[{\cal E}_2(\frac{\tau-1}{4})-{\cal E}_2(\frac{\tau-1}{2})\right]-2\theta_2(2\tau){\cal E}_2(\frac{\tau-1}{2})\\ \nn
A_1 &=& i\theta_2(2\tau)\left[{\cal E}_2(\frac{\tau-1}{4})-{\cal E}_2(\frac{\tau-1}{2})\right]+2\theta_3(2\tau){\cal E}_2(\frac{\tau-1}{2})\\
\end{eqnarray}

Since $A(\tau,z)$ is a Jacobi form of index 1 the question of sign changes reduces to a question on sign changes on the following two expressions for the 2A orbifold of $K3$ at magnetic charge $P^2=2$.
\begin{eqnarray}\label{sgnchng}
\theta_3(2\tau)\sum_{n=0}^{\infty} 24\sigma(4n+1)q^{n+1/4}-2\theta_2(2\tau)\left(1+\sum_{n=1}^{\infty} q^n\left( 24\sigma(2n)-48\sigma(n)\right)\right)\\ \nn
-\theta_2(2\tau)\sum_{n=0}^{\infty} 24\sigma(4n+3)q^{n+3/4}+2\theta_3(2\tau)\left(1+\sum_{n=1}^{\infty} q^n\left( 24\sigma(2n)-48\sigma(n)\right)\right)\\ \nn
\end{eqnarray}
We observe numerically that the only sign change appears in the second expression at $q^0$ upto a Fourier expansion of the first $1.6\times 10^5$ terms. The first few Fourier coefficients of $a_0$ and $a_1$ are listed below,
\begin{eqnarray}\nn
A_0(\tau)&=& 20 {q}^{1/4}+96 q^{5/4}+500 q^{9/4}+480 q^{13/4}+960 q^{17/4}+960 q^{21/4}+O(q^{25/4})\\ 
A_1(\tau) &=& 2-140 q-240 q^2-480 q^3-1100 q^4-1056 q^5-1440 q^6-O(q^7)
\end{eqnarray}
Using Propositions \ref{prop1p} and \ref{prop2p}, we argue that the average growth of the coefficients of $a_0$ and $a_1$ remain positive.
Let $C(x)$ be the average of the maximum growth of the coefficients of $\sigma(4x+k)$ such that $\sigma(4x+k)>2(\sigma(2x)-2\sigma(x))$ considering a frequency of occurrence $\frac{\kappa}{\log(x)}$ given by,
\begin{eqnarray}
C(M)=\kappa e^{\gamma} \frac{d}{dM}\int_{M_0}^M\frac{x\log\log(x)}{\log(x)}dx
\end{eqnarray}
and $D(x)$  be the growth of the rest of the coefficients which are of $O(x)$, we can write,
\begin{eqnarray}
D(M)=\frac{d}{dM}\int_{M_0}^M x\left(1-\frac{\kappa}{\log(x)}\right)dx
\end{eqnarray}
Then for $\kappa<2$ we have
$C(M)<D(M)$. Since the growth of theta functions which multiply $\sigma(4x+k)$ and $2(\sigma(2x)-2\sigma(x))$ are similar we expect the resultant average growth of these products to remain positive.

\subsection{Odd $N>2$}\label{oddn}
When we look at $N>2$ orbifolds there are in general more terms than the $N=2$ counterpart for the same magnetic charge $P^2=2$ whose positivity and growth we need to check individually.
In this part we describe the growth of coefficients corresponding to the mock Jacobi form $\psi_1^{(N)}$ where, $N$ is odd. These were evaluated in \cite{Chattopadhyaya:2018xvg}. We can write $\psi_1^{(N)}$ in terms of the coefficients of the Jacobi forms in the twisted elliptic genus of $K3$.
\begin{eqnarray}\label{p8}
& & - g^{(k+2)} \psi_1^{(N)} = \frac{N}{2}\alpha^{(1,0)}_{N} (\frac{N}{2}\alpha^{(1,0)}_{N}+1)\frac{A^2(N\tau,z)}{B(N\tau,z)}
+ \frac{N}{2}\alpha^{(1,0)}\frac{3}{16}E_4(N\tau)B(N\tau,z)
\\ \nn
&&\qquad\qquad +N^2\alpha^{(1,0)}_{N}
A(N\tau,z)\beta^{(1,0)}(N\tau)+\frac{N^2}{2}\beta^{(1,0)}_{N} (N\tau)^2 B(N\tau,z)\\ \nn
&&+\frac{N}{2}\left[4\frac{\theta_2(N\tau,z)^2}{\theta_2^2(N\tau)}\beta^{(1,0)}_{N} (2N\tau)+\frac{\theta_3(N\tau,z)^2}{\theta_3^2(N\tau)}\beta^{(1,0)}_{N} (\frac{N(\tau+1)}{2} )+\frac{\theta_4(N\tau,z)^2}{\theta_4^2(N\tau)}\beta^{(1,0)}_{N}
(\frac{N\tau}{2})\right].
\end{eqnarray}
We use the results obtained in \cite{Dabholkar:2012nd} to separate the Appell Lerch sum and the finite component
\begin{equation}\label{appel2}
 \frac{A^2 ( \tau, z) }{B(\tau, z)} = -9 \sum_{n \in \mathbb{Z}} \frac{ q^{n^2 +n} y^{2n +1}}{ (1- q^{n} y)^2 } +
\frac{1}{16}E_4(\tau) B(\tau, z)  + 18 {\cal H} (\tau, z) .
\end{equation}
Hence we can write the finite part after removing the fundamental string as follows:
\begin{eqnarray}\label{p8f}
& & - g^{(k+2)} \psi_1^{(N)} =\left( \frac{N^2}{4}(\alpha^{(1,0)}_{N})^2+ (\frac{N}{8}\alpha^{(1,0)}_N)\right)\left(\frac{1}{16}E_4(N\tau)B(N\tau,z)\right)
\\ \nn
&&+9N\alpha^{(1,0)}_{N} (\frac{N}{2}\alpha^{(1,0)}_{N}+1){\cal H}(N\tau,z) +N^2\alpha^{(1,0)}_{N}
A(N\tau,z)\beta^{(1,0)}(N\tau)+\frac{N^2}{2}\beta^{(1,0)}_{N} (N\tau)^2 B(N\tau,z)\\ \nn
&&+\frac{N}{2}\left[4\frac{\theta_2(N\tau,z)^2}{\theta_2^2(N\tau)}\beta^{(1,0)}_{N} (2N\tau)+\frac{\theta_3(N\tau,z)^2}{\theta_3^2(N\tau)}\beta^{(1,0)}_{N} (\frac{N(\tau+1)}{2} )+\frac{\theta_4(N\tau,z)^2}{\theta_4^2(N\tau)}\beta^{(1,0)}_{N}
(\frac{N\tau}{2})\right].
\end{eqnarray}
We can now use the following results relating the theta functions and the weak Jacobi forms $A,B$
\begin{eqnarray}
4\frac{\theta_2(\tau,z)^2}{\theta_2^2(\tau)}&=&\frac{4}{3}A(\tau,z)-\frac{2}{3}B(\tau,z){\cal E}_2(\tau),\\ \nn
4\frac{\theta_4(\tau,z)^2}{\theta_4^2(\tau)}&=&\frac{4}{3}A(\tau,z)+\frac{1}{3}B(\tau,z){\cal E}_2(\tau/2),\\ \nn
4\frac{\theta_3(\tau,z)^2}{\theta_3^2(\tau)}&=&\frac{4}{3}A(\tau,z)+\frac{1}{3}B(\tau,z){\cal E}_2((\tau+1)/2).
\end{eqnarray}
Using propositions \ref{prop1p} and \ref{prop2p} and the analysis of the section \ref{neq2} we observe that the only inconvenient piece may appear from the third line as follows:
$$-\frac{N}{3}B(\tau,z){\cal E}_2(\tau)\beta^{(1,0)}_{N} (2N\tau)$$

\noindent
Next we observe that the following lemma holds:
\begin{lem}
The coefficients of the Fourier expansion of $$\sum_{j=0}^1{\cal E}_N(\frac{\tau+j}{2}){\cal E}_2(\frac{\tau+j}{2})-16{\cal E}_N(2\tau){\cal E}_2(N\tau)$$
are always positive for positive powers of $q$ for odd $N$.
\end{lem}

\begin{proof}
This can be proved using $\sigma(nk)>k\sigma(n)\forall k,n$.
\end{proof}

\noindent
Hence the potential effect of the coefficients of the term $-\frac{N}{3}B(\tau,z){\cal E}_2(\tau)\beta^{(1,0)}_{N} (2N\tau)$ which could violate the sign conjecture are cancelled by the terms:
\begin{eqnarray}
\frac{N}{24}\left[{\cal E}_2\left(\frac{N(\tau+1)}{2}\right)\beta^{(1,0)}_{N} \left(\frac{N(\tau+1)}{2} \right)+{\cal E}_2 \left(\frac{N\tau}{2}\right)\beta^{(1,0)}_{N}\left(\frac{N\tau}{2}\right)\right]
\end{eqnarray}
for the Eisenstein series components for prime $N$. This is also true for composite odd $N$ the relevant case being the orbifold 15A provided the Fourier coefficient of the Eisenstein series component of $\beta_{15}^{(1,0)}$ remain positive.

\vspace{0.3cm}
\noindent
 For $N=11,15,23$ there are additional cusp forms in $\beta_{N}^{(1,0)}$ which grow as $ O(n^{9/4})$ while the growth of weight 2 Eisnenstein series remain $n^{3}$, hence we can neglect the contributions after first few terms which were checked to be positive in \cite{Chattopadhyaya:2017ews}.
This shows that the sign of the coefficients remain correct as expected from computing the helicity trace index $-B_6$.

\section{Open questions}\label{undone}
In this paper we only dealt with the cases of $N=2$ and odd $N$ at magnetic charge $P^2=2$. With the increasing magnetic charge the space of the modular forms also increase. For the cases with even $N>2$ one needs to look at the terms case by case as the generic result for $\psi^{(N)}_1$ involves several Eisenstein series of level 2. Individually looking at each term it is not clear whether they remain positive as in the case for $N$ odd. We leave further analysis for a future project.

For a complete proof of growth and signs of the Fourier coefficients of half integral weight Eisenstein series one requires more understanding on the behavior of the characters \cite{shurman,koblitz1984} and the behavior of $L$ functions \cite{ben-lola} associated with these objects. Unlike where the modular forms have negative weight where the behavior of the generating function of partitions have the dominating effect, the growth of coefficients of positive weight modular forms are understood less. However growth, sign changes are more well studied in the context of cusp forms \cite{kohnen13}\footnote{See also several other papers by the authors.}. There are interesting congruences on the behavior of coefficients of these forms and also for the class number generating function which is mock modular. As the level $N$ increases the basis of weight $k$ modular forms also increases which increases the complexity of the problem regarding frequency of possible sign changes. It is worth noting that not only we need $\sigma(4n+1)$ (or $\sigma(4n+3)$) to have few factors we simultaneously need $\sigma(n)$ to have a large number of factors such that the ratio in (\ref{prop1}) holds. So the contest is between ``highly composite numbers" and numbers which have relatively small number of factors. However we need to modify the definition of highly composite numbers through the statement of Proposition \ref{prop1p} instead of its usual definition \cite{ramanujan} as we observe numerically that the frequency of such terms we get is higher than the usual definition of ``highly composite". We leave to study these in future work.

\acknowledgments
The author thanks Kathrin Bringmann, Soumya Das for useful insights and references. The author also thanks Justin David, Jan Manschot and Ashoke Sen for useful discussions. The author is funded by fellowship from Dublin Institute for Advanced Studies. We thank International Center for Theoretical Sciences (ICTS) Bangalore for hospitality where part of the project was completed during an academic visit.

\appendix
\section{Notations and identities}
In this section we write the notations used in this paper and some identities.
The following conventions are used in the paper for theta and eta functions.
\begin{eqnarray}\nn
\theta_1(\tau,z) &=& \sum_{n}(-1)^{n-1/2}q^{\frac{(n-1/2)^2}{2}}y^{n-1/2}, \quad  \theta_2(\tau,z) =\sum_{n}q^{\frac{(n-1/2)^2}{2}} y^{n-1/2},\\ \nn
\theta_3(\tau,z) &=& \sum_{n}q^{\frac{n^2}{2}}y^{n}, \quad  \theta_4(\tau,z) = \sum_{n}(-1)^n q^{\frac{n^2}{2}}y^{n},\\ \nn
\eta(\tau)&=& q^{1/24}\prod_{n=1}^{\infty}(1-q^n), \qquad q=e^{2\pi i\tau},\; y=e^{2\pi iz}.
\end{eqnarray}
For $z=0$ we write $\theta_j(\tau,z)=\theta_j(\tau)$.

\noindent
Jacobi triple product identity used to extract the positive weight modular forms from the black hole partition functions is given as:
\begin{eqnarray}
\theta_2(\tau)\theta_3(\tau)\theta_4(\tau)=2\eta^3(\tau).
\end{eqnarray}

\noindent
Theta decomposition of Jacobi forms $A,B$ are done using the following double angle relations:
\begin{eqnarray}
\theta_1^2(\tau,z)&=&\theta_2(2\tau)\theta_3(2\tau,2 z)-\theta_3(2\tau)\theta_2(2\tau,2 z)\\ \nn
\theta_2^2(\tau,z)&=&\theta_2(2\tau)\theta_3(2\tau,2 z)+\theta_3(2\tau)\theta_2(2\tau,2 z)\\ \nn
\theta_3^2(\tau,z)&=&\theta_3(2\tau)\theta_3(2\tau,2 z)+\theta_2(2\tau)\theta_2(2\tau,2 z)\\ \nn
\theta_4^2(\tau,z)&=&\theta_3(2\tau)\theta_3(2\tau,2 z)-\theta_2(2\tau)\theta_2(2\tau,2 z).
\end{eqnarray}

\noindent
Relations between theta functions and Eisenstein series used in this paper can be given by,
\begin{eqnarray}
\theta_2^2(\tau)\theta_4^2(\tau) &=& \left[-\frac{i}{6}E_2(\frac{\tau-1}{4})+\frac{i}{2}E_2(\frac{\tau-1}{2})-\frac{i}{3}E_2(\tau)\right]\\ \nn
\theta_4^4(\tau)-\theta_2^4(\tau) &=& 2E_2(\tau)-E_2(\frac{\tau-1}{2})
\end{eqnarray}

\providecommand{\href}[2]{#2}\begingroup\raggedright\endgroup

\end{document}